\documentclass{svjour3}
\usepackage{mathrsfs}

\usepackage{amssymb,graphicx,amsmath}

\newcommand{\mcal}{\mathcal}

\newcommand{\msf}{\mathsf}

\newcommand{\R}{{\mathbb R}}

\newcommand{\tr}{\operatorname{tr}}
\newcommand{\tsp}{\tens{T}}

\def\tens#1{\ensuremath{\msf{#1}}}
\newcommand{\are}{\mcal{H}^{n-1}}

\newcommand{\grd}{\nabla}
\newcommand{\gr}[1]{\boldsymbol{#1}}
\newcommand{\wop}{W^{1,p}}

\newcommand{\gl}[1]{\tens{E}(#1)}

\newcommand{\skw}[1]{\operatorname{Skw}(#1)}

\renewcommand{\Omega}{\varOmega}
\renewcommand{\Gamma}{\varGamma}

\begin{document}

\title{A nonlinear Korn inequality based on the Green-Saint Venant
strain  tensor}

\author{Alessandro Musesti}
\institute{A. Musesti \at Dipartimento di Matematica e Fisica ``Niccol\`o Tartaglia'',
Universit\`a Cattolica del Sacro Cuore,
Via dei Musei 41, I-25121 Brescia (Italy)\\
\email{alessandro.musesti@unicatt.it}}

\date{\today}

\maketitle

\begin{abstract}
A nonlinear Korn inequality based on the Green-Saint Venant strain tensor is
proved, whenever the displacement is in the Sobolev space $W^{1,p}$,
$p\geq 2$, under Dirichlet conditions on a part of the boundary. The
inequality can be useful in proving the coercivity of a nonlinear
elastic energy.
\end{abstract}

\section{Introduction}

Korn inequality is one of the pillars of linear elasticity and it is
well-known for more than a century. In the most classical version it
writes
\begin{equation}
\|\grd\gr{u}\|_2\leq c\|e(\gr{u})\|_2
\end{equation}
provided that the displacement $\gr{u}$ vanishes on a sufficiently
large part of the boundary. Here $e(\gr{u})$ denotes the symmetric
part of $\grd\gr{u}$. Since $e(\gr{u})$ is the measure of the strain
in linear theories, Korn inequality provides a control of the deformation
gradient by means of the strain. The $L^p$ version of the inequality,
namely
\begin{equation}
\label{eq:kornp}
\|\grd\gr{u}\|_p\leq c\|e(\gr{u})\|_p
\end{equation}
if $\gr{u}$ vanishes on a part of the boundary,
followed more recently by general results about singular integral
operators, see \cite{John72}. 

However, in nonlinear elasticity a typical measure of the strain is given by
the so-called Green-Saint Venant (or Green-Lagrange) strain tensor
\[
\gl{\gr{u}}=\frac 1 2 (\tens F^\tsp\tens F -\tens
I)=\frac 1 2(\grd\gr{u}+\grd\gr{u}^\tsp+\grd\gr{u}^\tsp\grd\gr{u}),
\]
where $\tens F=\grd\gr{u}$ is the deformation gradient and
$\tens{I}$ denotes the identity tensor (see for instance~\cite[Sect.\
1.8]{Cia88}). Hence for a hyperelastic material it is customary to
write the energy density as a function of $\gl{\gr{u}}$. Unfortunately,
such an assumption, although very reasonable from a mechanical
viewpoint, is a source of difficulties on the mathematical side: both
the weakly lower semicontinuity and the coercivity of the elastic
energy can be a hard task, or even fail to hold.

Specifically, in proving an existence result for a hyperelastic body subject to
external loads and mixed boundary conditions, the Direct Method of the
Calculus of Variations can be a very powerful tool. Therefore, even if
the elastic energy density is well-behaved (for instance if it is
smooth and polyconvex), it is crucial to have some coercivity conditions
in order to construct a {\em bounded} minimizing sequence. Hence it
is very important to control the norm of the displacement by the
energy. Ultimately, one needs to control the displacement by
the Green-Saint Venant tensor, that is, a nonlinear version of the Korn
inequality with boundary conditions. The inequality can be useful also
in other applications: see for instance~\cite{Mar11}, where a similar
tool allows to prove the existence of a minimizer for a
nonlinear elastic problem by the implicit function theorem.

In the present note we prove the inequality
\begin{equation}
\label{eq:nlkorn}
\|\grd\gr{u}\|_p^p\leq c\|\gl{\gr{u}}\|_{p/2}^{p/2}
\end{equation}
in the case of a bounded Lipschitz body and for $\gr{u}$ in the
Sobolev space $W^{1,p}$, $p\geq 2$, and
$\det(\grd\gr{u}(\gr{x})+\tens I)>0$.  The result holds under weak
assumptions on the boundary conditions, namely it is enough for the
displacement $\gr{u}$ to vanish on a part of the boundary with
positive surface measure. We notice that, with a perhaps more simple
notation, equation~\eqref{eq:nlkorn} can be rewritten as
\[
\|\tens F - \tens I\|_p^p\leq C\|\tens F^\tsp\tens F -\tens
I\|_{p/2}^{p/2},\qquad p\geq 2,\ \det\tens F>0.
\]
Interestingly, as observed by the anonymous reviewer,
\eqref{eq:kornp} does not follow from~\eqref{eq:nlkorn} by
linearization, since the exponent on the right-hand side is $p/2$.

A similar inequality has been proved in~\cite{CiaMar04} for the case
$p=2$. Here we generalize that result to the case $p\geq 2$. An
essential ingredient of the proof is the $L^p$ version of the
celebrated {\em geometric rigidity lemma} of~\cite{FJM02}, originally
stated for $L^2$, which has been generalized for $1<p<\infty$ in Conti
\& Schweizer~\cite[Sec.\ 2.4]{ConSch06}.

\section{Notation}

Given an $n\times n$-matrix $\tens F$, we will denote by $|\tens F|$ its
Frobenius norm, that is
\[
|\tens F|=\sqrt{\displaystyle\tr(\tens{F}^\tsp\tens{F})}=\sum_{i,j=1}^n F_{ij}^2,
\]
where $\tens{F}^\tsp$ denotes the transpose of the matrix $\tens{F}$.
We will denote by $SO(n)$ the set of rotations in $\R^n$, that is 
\[
\tens{R}\in SO(n)\quad\Leftrightarrow\quad
\tens{R}^\tsp\tens{R}=\tens{I}\text{\ and\ }\det\tens{R}=1,
\]

Let $\Omega$ be a bounded connected open subset of $\R^n$, $n\geq 2$,
with Lipschitz boundary. Concerning the displacement, we assume that
$\gr{u}\in\wop(\Omega;\R^n)$ with $2\leq p<\infty$, endowed with the
usual norm
\[
\|\gr{u}\|_{1,p}=\left(\|\gr{u}\|_{p}^p+\|\grd\gr{u}\|_{p}^p\right)^{1/p}. 
\]
We will consider homogeneous Dirichlet boundary conditions on a part of the
boundary: let $\Gamma\subset\partial\Omega$ be a subset of the
boundary of $\Omega$ such that
\[
|\Gamma|:=\are(\Gamma)>0,
\]
where $\are$ denotes the $(n-1)$-dimensional Hausdorff
measure. We introduce the set
\[
\wop_\Gamma:=\left\{\gr{u}\in\wop(\Omega;\R^n):\ \gr{u}(\gr{x})=\gr{0} 
\quad\text{$\are$-a.e.\ in $\Gamma$}\right\}
\]
endowed with the norm $\|\cdot\|_{1,p}$, where the equality has to be
understood in the sense of traces of Sobolev functions. Notice that
$\wop_\Gamma$ is a closed linear subspace of $\wop(\Omega;\R^n)$.

A crucial lemma is the following (see also~\cite[Lemma 1]{CiaMar04}).
\begin{lemma}
\label{lem:ciamar}
There is a constant $c>0$ such that
\begin{equation}
\label{ciamar}
\operatorname{dist}(\tens F,SO(n))\leq c|\tens F^\tsp\tens F-\tens
I|^{1/2}
\end{equation}
for any matrix $\tens F$ with $\det \tens F>0$.
\end{lemma}
\begin{proof}
Since $\det\tens{F}>0$, it is a well-known consequence of the Polar
Factorization Theorem (see~\cite[Theorem 3.2-2]{Cia88}) that
\[
\operatorname{dist}(\tens F,SO(n))=|(\tens F^\tsp\tens
F)^{1/2}-\tens{I}|.
\]
Moreover, denoting by $v_1,\dots,v_n$ the eigenvalues of the matrix
$(\tens{F}^\tsp\tens F)^{1/2}$, by the equivalence of the Frobenius
and the spectral norm it follows that
\[
c_1\max_{1\leq i\leq n}|v_i-1|\leq
|(\tens F^\tsp\tens F)^{1/2}-\tens I|\leq 
c_2\max_{1\leq i\leq n} |v_i-1|
\]
for some $c_1,c_2>0$. 
Since $v_1,\dots,v_n\geq 0$, one has
\[
|(\tens F^\tsp\tens F)^{1/2}-\tens I|\leq c_2\max_{1\leq i\leq n} |v_i-1|\leq
c_2\max_{1\leq i\leq n} |v_i^2-1|^{1/2}=\frac{c_2}{c_1}|\tens F^\tsp\tens F-\tens
I|^{1/2}
\]
and the proof is complete.
\end{proof}

\section{The nonlinear Korn inequality}

We begin by proving the following lemma, which can also be found
in~\cite[Lemma 3.3]{DalNegPer02}.
\begin{lemma}
\label{lem:dal}
Let $\overline{K}$ denote the closure of the cone
\[
K=\Big\{t(\tens{I}-\tens{R}):\ t>0,\ \tens{R}\in SO(n)\Big\}.
\]
Then there exists $c>0$ such that
\[
\forall \tens{F}\in \overline{K}:\quad
|\tens F|\leq c \min_{\gr{z}\in\R^n}\left(\int_\Gamma |\tens{F}\gr{x}-\gr{z}|^2
\,d\are(\gr{x})\right)^{1/2}.
\]
\end{lemma}

\begin{proof}
First of all we notice that 
\begin{equation}
\label{eq:cone}
\overline{K}=K\cup \skw{n},
\end{equation}
where $\skw{n}$ denotes the set of 
skew-symmetric tensors of order $n$.
Indeed, assume that $\tens{F}\in\overline{K}\setminus K$ with
$\tens{F}\neq\tens{0}$; then there is an
unbounded sequence $(t_h)$ with $t_h>0$ such that 
\[
t_h(\tens{I}-\tens{R}_h)\to \tens{F}\quad\text{as $h\to\infty$},
\]
where $\tens{R}_h\in SO(n)$. In particular, $|\tens{I}-\tens{R}_h|\sim
1/t_h$ as $h\to\infty$. Representing an orthogonal tensor as
the exponential of a skew-symmetric tensor, we can write
\[
\tens{R}_h=\exp\left(t_h^{-1} \tens{A}\right)+o(t_h^{-1}) 
\quad\text{as $h\to\infty$},
\]
for some $\tens{A}\in\skw{n}$. Then we have
\[
\tens{F}=\lim_{h\to\infty}t_h\left[\tens{I}-\exp\left(t_h^{-1}
  \tens{A}\right)-o(t_h^{-1}) \right]=
-\frac{d}{dt}\exp(t\tens{A})\Big|_{t=0}=-\tens{A},
\]
hence $\tens{F}$ is skew-symmetric.

By~\eqref{eq:cone} it follows that 
\[
\tens{F}\in\overline{K},\, \tens{F}\neq
0\quad\Rightarrow\quad\dim\ker\tens{F}\leq n-2.
\]
Moreover it is easy to check that
\[
\int_\Gamma |\tens{F}\gr{x}-\gr{z}|^2\,d\are(\gr{x})\quad
\text{has minimum for}\quad \overline{\gr{z}}=\frac{1}{|\Gamma|}\int_\Gamma
\tens{F}\gr{x}\,d\are(\gr{x}).
\]

Now, assume by contradiction that for any $j\geq 1$ there exists
$\tens{F}_j\in \overline{K}$ such that $|\tens{F}_j|=1$ and
\[
\frac 1 j >\int_\Gamma|\tens{F}_j\gr{x}-\overline{\gr{z}}_j|^2\,d\are(\gr{x}),
\quad \overline{\gr{z}}_j=\frac{1}{|\Gamma|}\int_\Gamma
\tens{F}_j\gr{x}\,d\are(\gr{x}).
\]
Then there exists $\tens{F}\in \overline{K}$ such that $|\tens{F}|=1$ and
$\tens{F}_j\to \tens{F}$ up to a subsequence. By continuity it follows
that 
\[
\int_\Gamma|\tens{F}\gr{x}-\overline{\gr{z}}|^2\,d\are(\gr{x})=0,
\quad \overline{\gr{z}}=\frac{1}{|\Gamma|}\int_\Gamma
\tens{F}\gr{x}\,d\are(\gr{x}),
\]
whence $\tens{F}\gr{x}=\overline{\gr{z}}$ for $\are$-a.e. $\gr{x}\in\Gamma$.

Since $\Gamma$ has positive measure and $\Omega$ is Lipschitz, then
there exist at least $n-1$ pairs $(\gr{x}_i,\gr{y}_i)$ such that
$\gr{x}_i,\gr{y}_i\in\Gamma$ and $(\gr{x}_i-\gr{y}_i)$ are linearly independent. In
particular, $\tens{F}(\gr{x}_i-\gr{y}_i)=0$, hence $\dim\ker\tens{F}\geq n-1$, a
contradiction.
\end{proof}

Now we are ready to state and prove the main theorem.
\begin{theorem}
\label{thm:nonlK}
Let $\Omega$ be a bounded connected open subset of $\R^n$ with
Lipschitz boundary and $2\leq p <\infty$. Let
$\Gamma\subset\partial\Omega$ be a subset of the boundary of $\Omega$
such that $\are(\Gamma)>0$.
 
Then there exists a constant $c>0$ such that
\[
\|\grd\gr{u}\|_p^p\leq c\|\gl{\gr{u}}\|_{p/2}^{p/2}
\]
for every $\gr{u}\in W^{1,p}_\Gamma$ such that
$\det(\grd\gr{u}(\gr{x})+\tens I)>0$ for a.e.\ $\gr{x}\in\Omega$.
\end{theorem}

\begin{proof} 
The continuity of the trace operator gives
\begin{equation}
\label{eq:conttr}
\int_\Gamma |f(\gr{x})|^p\,d\are(\gr{x})\leq c_1\int_\Omega (|f(\gr{x})|^p+|\grd f|^p)\,d\gr{x}
\end{equation}
for every $f\in\wop(\Omega)$, where $c_1>0$ is some constant depending
on $\Omega$ and $\Gamma$.

Let $\gr{u}\in\wop_\Gamma$ and consider the vector field
$\gr{\phi}\in\wop(\Omega;\R^n)$ defined as $\gr{\phi}(\gr{x})=\gr{u}(\gr{x})+\gr{x}$.
For every $\tens R\in SO(n)$ define
\[
\overline{\gr{z}}_{\tens R}:=\frac{1}{|\Omega|}\int_\Omega (\gr\phi(\gr{x})-\tens R \gr{x})\,d\gr{x}.
\]
By~\eqref{eq:conttr} and Poincar\'e-Wirtinger inequality one has
\begin{equation}
\label{eq:pw}
\int_\Gamma |\gr\phi(\gr{x})-\tens R \gr{x} -\overline{\gr{z}}_{\tens R}|^p\,d\are(\gr{x})
\leq c_2 \int_\Omega |\grd\gr\phi(\gr{x})-\tens R|^p\,d\gr{x}.
\end{equation}
Notice that in the left-hand side one can replace $\gr\phi(\gr{x})$ with
the identity, in view of the boundary conditions. Moreover, by Lemma~\ref{lem:dal} it
follows that
\begin{equation}
\label{eq:ADD}
|\tens I-\tens R|\leq c_3 \left(\min_{\gr{z}\in\R^n}\int_\Gamma|
(\tens I-\tens R)\gr{x}-\gr{z}|^2\,d\are(\gr{x})\right)^{1/2}.
\end{equation}
In particular, choosing $\gr{z}=\overline{\gr{z}}_{\tens R}$, by
H\"older inequality (with $p\geq 2$) and~\eqref{eq:pw} it follows that
\[
|\tens I-\tens R|^p\leq c_3^p|\Gamma|^{\frac p 2-1}
\int_\Gamma|(\tens I-\tens R)\gr{x}-\overline{\gr{z}}_{\tens R}|^p\,d\are(\gr{x})
\leq c_4\int_\Omega |\grd\gr\phi(\gr{x})-\tens R|^p\,d\gr{x},
\]
hence
\[
\|\tens I-\tens R\|_p^p=|\Omega||\tens I-\tens R|^p\leq
c_5 \|\grd\gr\phi-\tens R\|_p^p.
\]
Then one has
\begin{equation}
\label{eq:boundR}
\|\grd\gr{\phi}-\tens I\|_p\leq
\|\grd\gr{\phi}-\tens R\|_p+\|\tens I-\tens R\|_p\leq
c_6\|\grd\gr{\phi}-\tens R\|_p
\end{equation}
for every $\tens R\in SO(n)$.

Now we need the celebrated
\emph{geometric rigidity lemma} by Friesecke, James \&
M{\"u}ller~\cite{FJM02}, which holds for $1<p<\infty$, as pointed out
by Conti \& Schweizer~\cite[Sec.\ 2.4]{ConSch06}:

\noindent\emph{there is a constant $K>0$ such that
\[
\min_{\tens R\in SO(n)} \|\grd\gr{\phi}-\tens R\|_p^p\leq K\int_\Omega
\operatorname{dist}^p(\grd\gr{\phi}(\gr{x}),SO(n))\,d\gr{x} 
\] 
for any $\gr{\phi}\in \wop(\Omega;\R^n)$.} 

By assumption, $\det\grd\gr{\phi}>0$ a.e.\ in $\Omega$, hence,
combining Lemma~\ref{lem:ciamar} with the geometric rigidity lemma it
follows that there exists $\tens R\in SO(n)$ such that
\[
\|\grd\gr{\phi}-\tens R\|_p^p\leq cK\int_\Omega
|\tens \grd\gr{\phi}(\gr{x})^\tsp\grd\gr{\phi}(\gr{x})-\tens I|^{p/2}\,d\gr{x}=
c_7\|\tens \grd\gr{\phi}^\tsp\grd\gr{\phi}-\tens I\|_{p/2}^{p/2}
\] 
and, by~\eqref{eq:boundR},
\[
\|\grd\gr{u}\|_p^p=\|\grd\gr{\phi}-\tens I\|_p^p\leq
c_6^p\|\grd\gr{\phi}-\tens R\|_p^p\leq
c_6^pc_7\|\tens \grd\gr{\phi}^\tsp\grd\gr{\phi}-\tens I\|_{p/2}^{p/2}
=c\|\gl{\gr{u}}\|_{p/2}^{p/2},
\]
which concludes the proof.
\end{proof}

\begin{remark}
In proving the previous theorem, we also proved the following consequence
of the $L^p$ version of the geometric rigidity lemma:

\noindent\emph{there is a constant $C>0$ such that
\[
\|\grd\gr{\phi}-\tens{I}\|_p\leq 
C\|\operatorname{dist}(\grd\gr{\phi},SO(n))\|_p 
\] 
for every $\gr{\phi}\in \wop(\Omega;\R^n)$ with $\gr{\phi}(\gr{x})=\gr{x}$ on $\Gamma$.} 
\end{remark}

\begin{acknowledgements}
  This research is partially supported by GNFM (Gruppo Nazionale di
  Fisica Matematica) of INdAM (Istituto
  Nazionale di Alta Matematica).
\end{acknowledgements}

\end{document}